\documentclass{article}
\usepackage[utf8]{inputenc}
\usepackage{amsmath, amsfonts, amssymb}
\usepackage{hyperref} 
\usepackage{tabularx}
\usepackage{graphicx}
\usepackage{dcolumn}
\usepackage{bm}
\usepackage{array}
\usepackage{amsthm}
\usepackage{wasysym}
\usepackage{authblk}

\usepackage[backend=biber, 
            style=numeric-comp,
            sorting=nyt,
            sortcites=true,
            maxbibnames=99]{biblatex}
            \renewbibmacro{in:}{}

\usepackage[T1]{fontenc}        
\usepackage{csquotes}           
\usepackage{pdfcomment}         

\DeclareCiteCommand{\mycite}
  {\usebibmacro{prenote}}
  {%
    \pdftooltip{\printnames{labelname}~(\printfield{year})}%
               {\printfield{title}. \printnames{author}. \printfield{journaltitle} \printfield{year}.}%
  }
  {\multicitedelim}
  {\usebibmacro{postnote}}
\usepackage{xcolor}

\theoremstyle{plain}
\newtheorem{theorem}{Theorem}[section]

\newtheorem{proposition}{Proposition}[section]

\theoremstyle{definition}
\newtheorem{definition}{Definition}[section]
\newtheorem{example}[theorem]{Example}
\newtheorem{note}[theorem]{Remark}

\newcommand{\R}{\mathbb{R}}
\newcommand{\Ld}{\mathcal{L}}
\newcommand{\bX}{{\bf X}}
\newcommand{\rk}{\:{\rm rk}\:}
\newcommand{\bEta}{{\bm \eta}}
\newcommand{\im}{{\mathrm{im}}}

\title{A $k$-contact Geometrical Approach to Pseudo-Gauge Transformation}
\author[1,4]{Mykhailo Hontarenko}
\author[2]{Javier de Lucas}
\author[2,3,5]{Adam Maskalaniec}

\affil[1]{\small Institute of Theoretical Physics, Faculty of Physics, Astronomy, and Applied Computer Science, Jagiellonian University, ul. Łojasiewicza 11, 30-348 Kraków, Poland}
\affil[2]{\small Faculty of Physics, University of Warsaw, ul. Pasteura 5, 02-093 Warszawa, Poland}
\affil[3]{\small Institute of Mathematics, Polish Academy of Sciences, ul. Śniadeckich 8, 00-656 Warszawa, Poland}
\affil[4]{\small Doctoral School of Exact and Natural Sciences, Jagiellonian University, ul. Łojasiewicza 11, 30-348 Kraków, Poland}
\affil[5]{\small Doctoral School of Exact and Natural Sciences, University of Warsaw, ul. Stefana Banacha 2c, 02-097 Warsaw, Poland}

\begin{document}

\maketitle

\begin{abstract}
We propose a starting point to the geometric description for the pseudo-gauge ambiguity in relativistic hydrodynamics, showing that it corresponds to the freedom to redefine the thermodynamic equilibrium state of the system. To do this, we develop for the first time a description of a relativistic hydrodynamic-like theory using $k$-contact geometry. In this approach, thermodynamic laws are encoded in a $k$-contact form, thermodynamical states are described via $k$-contact Legendrian submanifolds, and conservation laws emerge as a consequence of Hamilton-de Donder-Weyl (HdDW) equations. The inherent non-uniqueness of these solutions is identified as the source of the pseudo-gauge freedom. We explicitly demonstrate how this redefinition of equilibrium works in a model of a Bjorken-like expansion, where a pseudo-gauge transformation is shown to leave the physical dissipation invariant.
\end{abstract}

\section{Introduction}

Pseudo-gauge freedom, or pseudo-gauge ambiguity, has recently gained significant attention in physics, particularly in the context of relativistic hydrodynamics, spin-hydrodynamics, and magnetohydrodynamics, which are used to describe heavy-ion collisions. The energy-momentum tensor, which encapsulates the dynamics of physical systems, is known to be defined only up to the derivative of an antisymmetric tensor. This non-uniqueness was first addressed by Belinfante in 1939 \cite{Be39} and later by Rosenfeld in 1940 \cite{Ro40}. A broader analysis of this issue was presented in \cite{He76}, where it was explored in the context of spinning matter within gravitational theory. 

While pseudo-gauge transformations (PGTs) have been studied extensively in a wide range of classical field theories \cite{DFHR24,DFJR23} and quantum formulations of field theory (QFT) \cite{BT12,Na12,FP20,LSY20,Bu21, DFRS21,WWS22,DFJR23,BH25}, the results often exhibit a dependence on the choice of pseudo-gauge. This variability stems from the fact that proper pseudo-gauge transformations have yet to be rigorously and universally defined. 

Pseudo-gauge freedom extends beyond the energy-momentum tensor, affecting other conserved currents, such as the baryon current and the spin tensor. This freedom also influences the entropy current \cite{BDS23}, a critical component for studying dissipative corrections in hydrodynamic systems out of equilibrium. The core of the problem lies in the definition of thermal equilibrium itself. Pseudo-gauge transformations introduce gradient-dependent terms that are often associated with dissipation, blurring the line between a non-dissipative equilibrium state and a dissipative one. A proper definition of equilibrium should be robust against such ambiguities, yet no consensus has been reached on how to resolve this non-uniqueness, particularly for out-of-equilibrium configurations, too.

In this work, we propose a novel geometric origin for the pseudo-gauge ambiguity.  To achieve this, we develop, for the first time, a description of a relativistic hydrodynamic-like theory using the formalism of $k$-contact geometry. $k$-Contact geometry appeared recently as a generalization of contact geometry to deal with dissipative field theories \cite{RIV_23a,LRS24,GGM+_21,GRR_22a,SF_25,RSS_24a}. In this work $k$-contact geometry provides a natural setting where thermodynamics and dynamics are unified: the fundamental thermodynamic relations are encoded in Legendrian submanifolds of a $k$-contact manifold, while the conservation laws emerge as a consequence of the associated Hamilton-de Donder-Weyl (HdDW) equations. Then, the pseudo-gauge transformations, understood as mappings between different solutions of the HdDW equations, can either preserve the Legendrian submanifolds corresponding to chosen sets of equilibrium states or change the chosen set of equilibrium states.

Briefly speaking, our central result is the identification of the inherent non-uniqueness of solutions to the HdDW equations as the mathematical source of the pseudo-gauge freedom. We demonstrate that different solutions correspond to different choices of Legendrian submanifolds, which we interpret as different, yet physically equivalent, definitions of the equilibrium state. This perspective offers a unified framework that can be extended to more complex systems, such as magnetohydrodynamics or spin hydrodynamics, by incorporating additional degrees of freedom. We explicitly demonstrate our approach with a model of a Bjorken-like expansion, showing that a pseudo-gauge transformation is absorbed by this redefinition of equilibrium, leaving the physical dissipation invariant. Mathematically, we use relevant physical applications of $k$-contact manifolds, as well as develop the theory of Legendrian submanifolds in $k$-contact geometry, to describe thermodynamical relativistic states. The latter mathematically concerns a redefinition of the standard notion of a $k$-contact Legendrian manifold, the analysis of the dimension of Legendrian submanifolds in polarised $k$-contact manifolds, their possible parametrisations via particular families of $k$-contact Darboux coordinates \cite{Ri21}, and their analysis for the first time via parametrizing $k$-functions, related to generating functions for contact Legendrian submanifolds for the case $k=1$.

The outline of this paper is as follows. Section 2 provides a concise overview of the essential concepts of $k$-contact geometry. Section 3 introduces the framework of $k$-contact Hamiltonian systems and the HdDW equations. In Section 4, we discuss and provide some new results on  Legendrian submanifolds in $k$-contact polarised manifolds, which are crucial for defining physical equilibrium states. Section 5 applies this formalism to develop a $k$-contact description of an extensive relativistic hydrodynamic-like theory. The formal analysis of the pseudo-gauge degrees of freedom from this geometric perspective is presented in Section 6. In Section 7, we provide a concrete physical interpretation, applying our framework to a Bjorken-like expansion. Finally, Section 8 summarizes our findings and discusses future research directions. We assume familiarity of the reader with contact geometry, or at least with differential geometry.

\section{Basics on {\it k}-contact geometry}

It turns out that thermodynamical systems as well as relativistic hydrodynamics can be described geometrically by means of $k$-contact geometry. This branch of differential geometry was originally developed to extend the description of dissipative Hamiltonian systems in classical mechanics, which is provided by contact geometry, to classical field theories. This section briefly introduces the main notions and fundamental results of $k$-contact geometry. For a more detailed introduction, see \cite{LRS24} and \cite[Chapter 7]{Ri21}.  

Let us introduce conventions that will be used in this work. Throughout the article, we assume Einstein summation convention unless otherwise stated. A {\it generalised subbundle} $D$ of a vector bundle $E$ over a manifold $M$ is a subset of $E$ such that $D\cap E_x=:D_x$ is a vector subspace of the fibre $E_x$ for all $x\in M$. A generalised subbundle $D$ is {\it smooth} if it can be locally spanned by a finite family of sections of the vector bundle $E$. The {\it rank} of a generalised subbundle $D$ at $x\in M$ is the dimension of the vector space $D_x$. The generalised subbundle $D$ is {\it regular} if it is  locally of constant rank. In this work, generalised subbundles in the tangent bundle are called {\it distributions}, while distributions in the cotangent bundle are called {\it codistributions}. If a vector field $X$ takes values in a distribution $D$, we write $X\in \Gamma(D)$. 

Every differential $p$-form on a manifold $M$ taking values in a vector space $V$, namely an element $\bm\alpha\in\Omega^p(M, V)=\Omega^p(M)\otimes V$ induces a distribution $\ker\bm\alpha$, which is the  kernel of the induced map $\hat{\bm\alpha}:TM\longrightarrow\Lambda^{p-1}T^*M\otimes V$ with $\hat{\bm \alpha}(X)=\iota_X{\bm \alpha}$. 

Given a distribution $D\subset TM$, the annihilator of $D$, which is denoted by $D^\circ$, is a codistribution. Similarly, the annihilator of a codistribution is a distribution.

\begin{definition}\label{def:kcontactform} A {\it $k$-contact form} on a manifold $M$ is an $\R^k$-valued differential form ${\bm \eta}\in\Omega^1(M,\R^k)$ such that:
\begin{enumerate}
    \item $\ker{\bm \eta}\subset TM$ is a regular non-zero distribution of corank $k$,\label{eq:Contact1}
    \item $\ker d{\bm \eta}\subset TM$ is a regular distribution of rank $k$,\label{eq:Contact2}
    \item $\ker{\bm \eta}\cap \ker{d\bm \eta}=0$.\label{eq:Contact3}
\end{enumerate}
The distribution $\ker{\bm \eta}$ is called a {\it $k$-contact distribution} while $\ker{d\bm \eta}$ is called the associated {\it Reeb distribution}. The pair $(M,{\bm \eta})$ is called, for simplicity\footnote{In a more precise manner, $(M,\bEta)$ should be called a {\it co-oriented $k$-contact manifold} \cite{LRS24}, but since we here deal only with co-oriented $k$-contact manifolds, the term `co-oriented' will be omitted.}, a {\it $k$-contact manifold}. From now on, we write  $\bm\eta=\eta^\alpha\otimes e_\alpha$ for a $k$-contact form on a $k$-contact manifold, where $\eta_1,\ldots,\eta_k\in\Omega^1(M)$ and $\{e_1,\ldots,e_k\}$ stand for the canonical basis of $\R^k$.
\end{definition}
There exists a more general notion of a $k$-contact manifold that relies only on the properties of the contact distribution and does not require the definition of the $k$-contact form \cite{LRS24}. It is also worth noting that there are other mathematical structures called $k$-contact structures \cite{Je01}. 

Let us prove a proposition that characterizes $k$-contact manifolds in terms of local coordinates.

\begin{theorem}\label{Th:Reeb}
    Let $(M,\bm\eta)$ be a $k$-contact manifold. Then, there exist vector fields $R_1,\ldots,R_k\in\mathfrak{X}(M)$, defined uniquely by
    \begin{equation}\label{eq:Reeb1}
        \iota_{R_\alpha}\eta^\beta=\delta_\alpha^\beta,\qquad \alpha,\beta=1,\ldots,k,
    \end{equation}
    \begin{equation}\label{eq:Reeb2}
        \iota_{R_\alpha}d\eta^\beta=0,\qquad \alpha,\beta=1,\ldots,k,
    \end{equation}
    called the {\it Reeb vector fields} of $(M,\bEta)$. The Reeb vector fields span the Reeb distribution and commute with each other, namely
    $$
    [R_\alpha,R_\beta]=0,\qquad \alpha,\beta=1,\ldots,k.
    $$
\end{theorem}
\begin{proof}As $\ker{\bm \eta}\oplus \ker{d\bm\eta}=TM$, while  $\ker d\bm\eta$ has constant rank equal to $k$, and $\ker \bEta$ has constant rank equal to $\dim M -k$, there exist uniquely defined vector fields $R_1,\ldots,R_k$ dual to $\eta^1,\ldots,\eta^k$ taking values in $\ker d\bEta$. Since $R_1,\ldots,R_k$ are dual to $\eta^1,\ldots,\eta^k$,  the vectors  fields $R_{1},\ldots,R_{k}$ become linearly independent at every point of $M$. As $\ker{d\bm\eta}$ has rank $k$, the vector fields $R_1,\ldots,R_k$ form a basis for $\ker{d\bm \eta}$. By conditions \eqref{eq:Reeb1} and \eqref{eq:Reeb2}, it follows that $\mathcal{L}_{R_\alpha}\bEta=0$. Then,   the identity $\iota_{[R_\alpha, R_\beta]}=\Ld_{R_\alpha}\iota_{R_\beta}-\iota_{R_\beta}\Ld_{R_\alpha}$ and the definition of Reeb vector fields as dual vector fields to $\eta^1,\ldots, \eta^k$ show that
    $$
    \iota_{[R_\alpha,R_\beta]}\bm\eta=0,\qquad\iota_{[R_\alpha,R_\beta]}d\bm\eta=0, 
    $$
and since $\ker \bEta\cap\ker d\bEta=0$, it follows that $[R_\alpha,R_\beta]=0$ for $\alpha,\beta=1,\ldots,k$.
\end{proof}

\begin{proposition}
    The Reeb distribution of a $k$-contact manifold $(M,\bm \eta)$  is integrable.
\end{proposition}
\begin{proof}
    If $X,Y\in\Gamma(\ker d \bm\eta)$, then
    $$
    \iota_{[X,Y]}d\bm\eta=\Ld_X\iota_Yd\bm\eta-\iota_Y\Ld_X d\bm\eta=-\iota_Y\Ld_X d\bm\eta=-\iota_Y(d\iota_X+\iota_Xd)\bm\eta=0
    $$
    and $[X,Y]\in \ker{d\bm \eta}$. Hence, the Reeb distribution is involutive. Moreover, since the Reeb distribution is regular, it is also integrable. 
\end{proof}
\begin{proposition}
    Every $n$-dimensional $k$-contact manifold $(M,\bm\eta)$ admits local coordinates $\{s^\alpha,x^I\}$, for $\alpha=1,\ldots,k$ and $I=1,\ldots,n-k$ such that
    $$
    R_\alpha=\frac{\partial}{\partial s^\alpha},\qquad \eta^\alpha=ds^\alpha-f^\alpha_I(x)dx^I,
    $$
    for some functions $f^\alpha_I$ depending only on the coordinates $x^I$.
\end{proposition}
\begin{proof}
    Since Reeb vector fields commute and are linearly independent at each point, there exist local coordinates $\{s^\alpha,x^I\}$ on an open $U\subset M$ that straighten Reeb vector fields simultaneously and $R_\alpha=\frac{\partial}{\partial s^\alpha}$ for $\alpha=1,\ldots,k$. Condition \ref{eq:Reeb1} in Theorem \ref{Th:Reeb} implies
    $$
    \eta^\alpha=ds^\alpha-f^\alpha_I(s,x)dx^I,\qquad\alpha=1,\ldots,k,
    $$
    for some functions $f^\alpha_I\in C^\infty(U)$.  Moreover, Condition \ref{eq:Reeb2} from Theorem \ref{Th:Reeb} yields $f^\alpha_I(s,x)=f^\alpha_I(x)$ for $\alpha=1,\ldots,k$ and every $I$.
\end{proof}

Contrary to the Reeb distribution $\ker d\bm \eta$, the {\it $k$-contact distribution}, namely $\ker{\bm \eta}$, is not integrable. Indeed, a distribution $D\subset TM$ defined by $D=\ker \bm\zeta$ for an $\mathbb{R}^k$-valued differential form $\bm\zeta\in\Omega^1(M,\R^m)$ of constant rank $m$ is integrable if and only if 
\begin{equation}\label{eq:DiffFormIntegrability}
    d\bm\zeta \lvert_{D\times D}(\cdot,\cdot)=-\bm\zeta([\cdot,\cdot])=0.
\end{equation}
But for a $k$-contact form $\bm\eta$, one has $\ker\bm\eta\cap\ker d\bm\eta=0$, which contradicts the above condition of integrability expressed via differential forms. However, we may look for submanifolds of $k$-contact manifolds whose tangent manifold is contained in the $k$-contact distribution. 

It turns out that under additional conditions imposed on the $k$-contact form, one may guarantee the existence of a more specific type of coordinates.
\begin{definition}
    Let $(M,\bEta)$ be a $k$-contact manifold of dimension $\dim M =n+kn+k$. A {\it polarization} of the $k$-contact distribution $\ker{\bm \eta}$ is an integrable subbundle $\mathcal{V}\subset \ker{\bm \eta}$ of rank $nk$. If such distribution exists, then $(M,\bm\eta,\mathcal{V)}$ is called a {\it polarized $k$-contact manifold}.
\end{definition}

\begin{theorem}\label{Th:DarbouxkContact}
    If $(M,\bm\eta,\mathcal{V})$ is a polarized $k$-contact manifold, then around every point of $M$ there exists, on an open neighbourhood $U$, a coordinate system $\{s^\alpha,q^i,p^\alpha_i\}$ , where $i=1,\ldots, n$ and $\alpha=1,\ldots,k$, such that
    $$
    \eta^\alpha\vert_U=ds^\alpha-p^\alpha_i dq^i,\qquad \alpha=1,\ldots,k.
    $$
\end{theorem}
The coordinates $\{s^\alpha,q^i,p^\alpha_i\}$ in Theorem \eqref{Th:DarbouxkContact} are called {\it $k$-contact Darboux coordinates} or simply Darboux coordinates if their $k$-contact nature is understood from context. The proof of Theorem  \ref{Th:DarbouxkContact} can be found in \cite[ Theorem 7.1.11, p. 100]{Ri21}.
\begin{example} (Canonical $k$-contact form) The manifold $M=\R^k\times \bigoplus_{\alpha=1}^k T^*Q$ admits a canonical $k$-contact form given by a family of one-forms
$$
\eta^\alpha=ds^\alpha-{\rm pr}_\alpha^*\theta,\qquad \alpha=1,\ldots,k,
$$
where $\theta$ is the Liouville form on $T^*Q$ and ${\rm pr}_\alpha$ denotes the  canonical projection from $M$ onto the $\alpha$-th factor of $T^*Q$ in $M$. The canonical symplectic Darboux coordinates on the cotangent bundle and the canonical coordinates on $\R^k$ induce $k$-contact Darboux coordinates $\{s^\alpha,p_i^\alpha,q^i\}$ on $\R^k\times \bigoplus_{\alpha=1}^k T^*Q$. As $d\eta^\alpha=dq^i\wedge dp^\alpha_i$ for $\alpha=1,\ldots,k$, the Reeb distribution is locally of the form
$$
\ker{d\bm \eta}=\left\langle\frac{\partial}{\partial s^1},\ldots,\frac{\partial}{\partial s^k}\right\rangle.
$$
Moreover, the $k$-contact distribution is locally of the form
$$
\ker{\bm \eta}=\left\langle p^\alpha_1\frac{\partial}{\partial s^\alpha}+\frac{\partial}{\partial q^1},\ldots,p^\alpha_n\frac{\partial}{\partial s^\alpha}+\frac{\partial}{\partial q^n}\right\rangle\oplus \mathcal{V},
$$
where
$$
\mathcal{V}=\left\langle\frac{\partial}{\partial p^\alpha_i}\right\rangle_{\substack{i=1,\ldots,n\\ \alpha=1,\ldots,k}}\subset \ker \bEta,
$$
is a polarization of the $k$-contact form $\bEta$. It is worth noting that there is no other larger integrable isotropic distribution containing $\mathcal{V}$ within $\ker \bEta$.
\end{example}

The interesting point of the polarization notion in the $k$-contact context is that it is an integrable isotropic distribution of maximal rank within the $k$-contact distribution $\ker\bEta$. 

\begin{note}
    In the case when $k=1$, the definition of a $k$-contact form is equivalent to the definition of a contact form. Therefore, all theorems on $k$-contact manifolds are also valid in the contact setting. Let us prove that the definition of the contact form is equivalent to the definition of the $k$-contact form for $k=1$.

    Let $(M,\eta)$ be a contact manifold, namely a $(2n+1)$-dimensional manifold $M$ endowed with a one-form $\eta$ such that $\eta\wedge (d\eta)^n$ is a volume form on $M$. Since $\eta\wedge (d\eta)^n$ is a volume form, $\eta$ is a non-vanishing form, so $\ker\eta$ is a regular distribution of rank $2n$. Moreover, the same condition yields that $\ker\eta\cap\ker d\eta=0$. Using the canonical form of a two-form (see \cite[Proposition 3.1.2]{AM78}), we see that the rank of the distribution $\ker d\eta$ is odd. As $\ker\eta\cap\ker d\eta=0$, the only possibility is that the rank of $\ker d\eta$ is equal to one, so $\ker\eta\oplus\ker d\eta=TM$.
    
    Let us prove the converse implication. Consider a 1-contact manifold $(M,\eta)$. Let us prove that $\eta$ is a contact form. By Definition \ref{def:kcontactform}, the distribution $\ker\eta$ has constant rank equal to $\dim M-1$. Moreover, $\ker d\eta$ has constant rank equal to $1$. Using the canonical form of a two-form (see \cite[Proposition 3.1.2]{AM78}) and $\ker\eta\cap\ker d\eta=0$, we conclude that $\dim M-1=2n$ for some $n\in \mathbb{N}$. Since $\ker{ \eta}\cap \ker{d \eta}=0$, it follows that $\ker\eta\wedge (d\eta)^{n}$ is not vanishing, so $\eta\wedge (d\eta)^n$ is a volume form on $M$. Thus, $\eta$ is a contact form.

\end{note}

\section{{\it k}-contact Hamiltonian systems}

$k$-Contact Hamiltonian systems allow us to describe the dynamics of dissipative field theories \cite{Ri21}. They were developed analogously to the Hamiltonian systems known from symplectic, contact, and $k$-symplectic geometry. To formulate $k$-contact Hamiltonian equations, we introduce $k$-vector fields.

Let us denote by
$ 
\pi:\oplus^kTM\longrightarrow M
$
the natural projection that endows the Whitney sum $\oplus^kTM$ with a structure of a vector  bundle over $M$ and let $
\pi_\alpha:\oplus^kTM\longrightarrow TM,$
be the natural projection onto the $\alpha$-th direct component.
\begin{definition}
    A {\it $k$-vector field} ${\bf X}$ is a section of the vector bundle $\oplus^kTM$ over  $M$.
\end{definition}
The vector space of $k$-vector fields on $M$ is denoted by $\mathfrak{X}^k(M)$. Each $k$-vector field ${\bf X}$ is uniquely defined by a family $(X_1,\ldots,X_k)$, with  $X_\alpha=\pi_\alpha\circ{\bf X}$ for $\alpha=1,\ldots,k$, and we may denote ${\bf X}=(X_1,\ldots, X_k)$, where $X_1,\ldots,X_k$ are usually called the {\it components} of $\bX$. Each $k$-vector field induces a distribution on $M$ spanned by its components. We define the inner product of a $k$-vector field with an $\R^k$-valued differential $p$-form as
\begin{align*}
    \iota:\mathfrak{X}^k(M)\times \Omega^p(M,\R^k)&\longrightarrow \Omega^{p-1}(M),\\
    ({\bf X},\omega^\alpha\otimes e_\alpha)&\longmapsto \iota_{\bf X}{\bm \omega}:=\iota_{X_\alpha}\omega^\alpha,
\end{align*}
The notion of the tangent lift of a curve in a manifold $M$ naturally extends to tangent lifts of mappings from $\mathbb{R}^k$ to $M$ as follows.
\begin{definition}
    The {\it first prolongation} of a map $\psi:\R^k\longrightarrow M$ takes the form $\psi':\R^k\longrightarrow \bigoplus^kTM,$ with
    \begin{equation}
        \label{eq:prolongation}\psi'(t^1,\ldots,t^k)=\left(T_{(t_1,\ldots,t_k)}\psi\left(\frac{\partial}{\partial t^1}\right),\ldots,T_{(t_1,\ldots,t_k)}\psi\left(\frac{\partial}{\partial t^k}\right)\right),
    \end{equation}
    where $\{t^1,\ldots,t^k\}$ are the canonical coordinates on $\R^k$.
\end{definition}
\begin{definition}
    An {\it integral section} of a $k$-vector field $\bf X$ is a map $\psi:\R^k\longrightarrow M$ that satisfies 
    $$
    \psi'={\bf X}\circ \psi.
    $$
    A $k$-vector field $\bX$ is {\it integrable} if any point of $M$ lies in the image of an integral section of $\bX$.
\end{definition}
\begin{proposition}
    A $k$-vector field $\bX=(X_1,\ldots,X_k)$ is integrable if and only if 
    \begin{equation}\label{eq:IntegrabilityCondition}
        [X_\alpha,X_\beta]=0,\qquad\alpha,\beta=1,\ldots,k.
    \end{equation}
\end{proposition}
\begin{proof}
    If $\bX$ is integrable and $\psi$ is an integral section of $\bX$, then equation \eqref{eq:prolongation} yields that each vector field $X_\alpha\in\mathfrak{X}(M)$ is $\psi$-related to the vector field $\frac{\partial}{\partial t^\alpha}\in\mathfrak{X}(\R^k)$. Therefore, $\left[\frac{\partial}{\partial t^\alpha},\frac{\partial}{\partial t^\beta}\right]=0$ and this implies 
    $$T\psi\left(\left[\frac{\partial}{\partial t^\alpha},\frac{\partial}{\partial t^\beta}\right]\right)=\left[T\psi\left(\frac{\partial}{\partial t^\alpha}\right),T\psi \left(\frac{\partial}{\partial t^\beta}\right)\right]=[X_\alpha|_{\psi}, X_\beta|_{\psi}]=0.$$ 
    Considering the above for all possible $\psi$, we obtain that condition \eqref{eq:IntegrabilityCondition} is necessary. It is also sufficient by the Fröbenius theorem.
\end{proof}

It is worth noting that ${\bf X}$ may have some integral sections despite not being integrable at every point. Let us now formulate the fundamental theory of $k$-contact Hamiltonian equations.

\begin{definition}
    A {\it $k$-contact Hamiltonian system} is a triple $(M,{\bm \eta},H)$, where $(M,{\bm \eta})$ is a $k$-contact manifold and $H\in C^\infty(M)$. The {\it geometric $k$-contact Hamilton-de Donder-Weyl (HdDW) equations} for $(M,\bm \eta,H)$ read
    \begin{equation}\label{eq:HdDW}
     \iota_{\bX}d{\bm\eta}=dH-(\mathcal{L}_{R_\alpha}H)\eta^\alpha, \qquad
         \iota_{\bX}{\bm\eta}=-H.
    \end{equation}
    Meanwhile, the {\it $k$-contact Hamilton-de Donder-Weyl (HdDW) equations} for $(M,\bm \eta,H)$ are the system of PDEs for the integral curves $\psi:\mathbb{R}^k\rightarrow M$ of $\bX$, namely
    \begin{equation}\label{eq:HdDW2}
     \iota_{\psi'}d{\bm\eta}=d(H\circ \psi)-[(\mathcal{L}_{R_\alpha}H)\circ\psi ]\eta^\alpha\circ \psi', \qquad
         \iota_{\psi'}{\bm\eta}=-H\circ \psi.
    \end{equation}
\end{definition}

It is worth stressing that geometric HdDW equations are equations on the unknown $\bX$. Meanwhile, HdDW equations are systems of partial differential equations on the integral sections of $\bX$, whose existence depend on the integrability properties of $\bX$. In a polarized $k$-contact manifold, $k$-contact Darboux coordinates allow us to write that HdDW equations  become
\begin{equation}\label{eq:k-contact-HDW-Darboux-coordinates} 
        \frac{\partial \psi^i}{\partial t^\alpha} = \frac{\partial h}{\partial p_i^\alpha}\circ\psi\,,\qquad
        \frac{\partial \psi^\alpha_i}{\partial t^\alpha} = -\left( \frac{\partial h}{\partial q^i} + p_i^\alpha\frac{\partial h}{\partial z^\alpha} \right)\circ\psi\,,\qquad
        \frac{\partial \psi^\alpha}{\partial t^\alpha} = \left( p_i^\alpha\frac{\partial h}{\partial p_i^\alpha} - h \right)\circ\psi\,,
\end{equation}
where $\psi^i=q^i\circ \psi$, $\psi^\alpha_i=p^\alpha_i\circ \psi$ and $\psi^\alpha=s^\alpha\circ \psi$.

Now, the following statement shows that $k$-contact Hamiltonian dynamics is a suitable framework for describing relativistic hydrodynamics, as it naturally incorporates pseudo-gauge degrees of freedom.
\begin{theorem}\label{thm:HdDWNonUniqueness}
    Geometric $k$-contact Hamilton-de Donder-Weyl  equations for a $k$-contact manifold $(M,\bEta,H)$ always admit solutions, which are unique only when $k=1$. In such a case, the $k$-contact Hamiltonian-de Donder-Weyl equations have also solutions.
\end{theorem}
\begin{proof}
    Any $k$-vector field $\bX$ on $M$ can be decomposed uniquely into $\bX=\bX^C+\bX^R$ for $\bX^C\in \Gamma(\ker\bm \eta)$ and $\bX^R\in \Gamma(\ker d\bm\eta)$.  Hence, one can write the geometric $k$-contact HdDW equations for $(M,\bEta,H)$ as 
    \begin{equation*}
\iota_{\bX^C}d{\bm\eta}=dH-(\mathcal{L}_{R_\alpha}H)\eta^\alpha, \qquad \iota_{\bX^R}{\bm\eta}=-H.
    \end{equation*}
    Since $dH-(\mathcal{L}_{R_\alpha}H)\eta^\alpha$ takes values in $(\ker d\bEta)^\circ$, geometric $k$-contact HdDW equations admit solutions for all $H\in C^\infty(M)$ if and only if vector bundle maps
    \begin{align*}
        \rho_1:\oplus^k\ker \bm \eta &\longrightarrow (\ker{d\bm \eta})^\circ,\\
        \bX^C&\longmapsto\iota_{\bX^C}d\bEta,
    \end{align*}    \begin{align*}
        \rho_2:\oplus^k \ker{d\bm \eta}&\longrightarrow M\times\R,\\
        \bX^R&\longmapsto\iota_{\bX^R}\bEta,
    \end{align*}
are surjective. Clearly, $\rho_2$ is surjective. Let us prove that $\rho_1$ is also surjective by reduction to absurd. If $\rho_1$ were not surjective, then the codistribution spanned by the $ d\eta^\alpha(w)$ with $w\in \ker\bEta$ and $\alpha=1,\ldots,k$, is strictly contained in $(\ker{d\bm \eta})^
\circ $ at least at a point $x\in M$. Hence, there exists a tangent vector $0\neq v\in \ker{\bm \eta}_x$ satisfying that $(d\eta^\alpha)_x(v,\cdot)=0$ for $\alpha=1,\ldots,k$ on $\ker \bm \eta_x$. Since $(d\eta^\alpha)_x(v,R_i)=0$ for $i=1,\ldots,k$, it follows that $v\in \ker d\bEta_x\cap \ker \bEta_x$. Thus, $\bm \eta$ is not a $k$-contact form. This is a contradiction and $\rho_1$ is surjective.

    Now, note also that 
    $$
    \rk\ker\rho_1=\rk(\oplus^k \ker{\bm \eta})-\rk\im\,\rho_1=(k-1)m,
\qquad    \rk\ker\rho_2=\rk(\oplus^k \ker{d\bm \eta})-\rk\im\rho_2=k^2-1,
    $$
    where we assume $\dim M=k+m$. Hence, solutions to geometric $k$-contact HdDW equations are unique only for ordinary contact manifolds. Since these are ordinary differential equations, they also admit solutions.
\end{proof}

It is worth noting that the $k$-vector fields satisfying the geometric $k$-contact HdDW equations do not need to be integrable. Hence, the determination of existence of solutions to their $k$-contact HdDW equations is an involved task in general.

\section{Legendrian submanifolds of {\it k}-contact manifolds}
In this section, we recall the notion of isotropic and Legendrian spaces and submanifolds of co-oriented $k$-contact manifolds, introduced in \cite{LRS24}. We prove that such Legendrian spaces are maximal with respect to inclusion in isotropic spaces. Moreover, we construct examples of Legendrian submanifolds of different dimensions for the same $k$-contact manifold, introducing parametrizing $k$-functions for Legendrian submanifold of $k$-contact manifolds, which are related to generating functions of Legendrian submanifold in contact geometry for $k=1$ (see \cite[p. 367]{Ar89}). These results will play an important role in the $k$-contact formulation of relativistic hydrodynamic-like theory in further sections. 

\begin{definition} 
  The {\it $k$-contact orthogonal} of a vector subspace $E_x\subset \ker \bm \eta_x$, for $x\in M$ and a $k$-contact manifold $(M,\bm \eta)$, is the vector subspace
    $$
    E_x^{\perp_{\bm\eta}}:=\{v_x\in \ker \bm \eta_x\;\vert\; d\bm\eta(v_x,w_x)=0,\,\forall w_x\in E_x\}.
    $$
    Then, $E_x$ is {\it isotropic} if $E_x\subset E_x^{\perp_{\bm\eta}}$. Meanwhile, $E_x$ is {\it Legendrian} if $E_x$ is isotropic and it admits a  complement $F_x\oplus E_x=T_xM$ such that $d\bm\eta\vert_{F_x\times F_x}=0$. A distribution $E\subset TM$ is {\it isotropic (resp. Legendrian)} if  every $E_x$, with $x\in M$, is isotropic (resp. Legendrian).
\end{definition}

It can be proved that $E_x$ is Legendrian if and only if it admits a complement $W_x\subset \ker \bEta$ such that $E_x\oplus W_x=\ker \bEta$ and $W_x$ is isotropic relative to $\bEta$. Indeed, if $W_x$ exists, then $E_x\oplus(W_x\oplus\ker d\bEta_x)=T_xM$ and $F_x=W_x\oplus \ker d\bEta_x$ satisfies that $d\bEta|_{F_x\times F_x}=0$. On the other hand, if $E_x\oplus F_x=T_xM$ and $F_x$ is such that $d\bEta|_{F_x\times F_x}=0$, then one has that $F_x+\ker \bEta_x\supset F_x\oplus E_x=T_xM$. Moreover,
\begin{align*}
    &\dim F_x\cap\ker \bEta_x=\dim F_x+\dim \ker \bEta_x-\dim (F_x+\ker \bEta_x)=\\
    &\dim M-\dim E_x+\dim \ker \bEta_x-\dim M=\dim \ker\bEta_x-\dim E_x.
\end{align*}Hence, $W_x=F_x\cap \ker \bEta_x$ is a complement to $E_x$ in $\ker\bEta_x$ and it is isotropic relative to $\bEta$ because it is contained on a space, $F_x$, where $d\bEta|_{F_x\times F_x}=0$. 
The latter shows that, instead of the original definition in \cite{LRS24}, one can define Legendrian subspaces using only the orthogonal relative to $\bEta$.

Every isotropic distribution is integrable by condition \eqref{eq:DiffFormIntegrability}. Note also that the notion of being $k$-contact orthogonal does not really depend on the $k$-contact form $\bm \eta$ as long as their kernel is the same \cite{LRS24}. 
\begin{definition}
    A submanifold $N$ of a $k$-contact manifold $(M,\bm\eta)$ is {\it isotropic (resp. Legendrian)} if, for every $x\in N$, the tangent space $T_xN$ is isotropic (resp. Legendrian). 
\end{definition}
It is worth stressing that $k$-contact  manifolds may admit Legendrian submanifolds of different dimensions. However, all Legendrian submanifolds are maximal relative to their inclusion in other isotropic submanifolds as proved next. 

\begin{theorem}\label{Th:MaxLeg}
    If $L$ is a Legendrian submanifold of a $k$-contact manifold $(M,\bm\eta)$, then there exists no isotropic submanifold $L'$ such that $L
    \subset L'$ and $\dim L'>\dim L$.
\end{theorem}
\begin{proof} Let us prove this by contradiction. Assume that $L$ is a Legendrian submanifold and $L'$ is an isotropic submanifold with $L\subset L'$. For $x\in L$, there exists a subspace $W_x\subset \ker \bEta_x$ so that
    $$
    T_xL\oplus W_x=\ker\bEta_x.
    $$
    As $L$ is Legendrian, $W_x$ can be chosen to be isotropic relative to $\bEta$. Since $\dim T_xL'\geq\dim T_xL$ by assumption, $T_xL'\cap W_x\neq 0$. Moreover, $d\bEta\vert_{W_x\times W_x}=0$ implies that $T_xL'\cap W_x\subset W_x^{\perp_\bEta}$. Since $L\subset L'$ and $L'$ is isotropic, one has that  $d\bEta\vert_{T_xL'\times T_xL'}=0$ yields that $d\bEta\vert_{T_xL\times T_xL'}=0$, what can be rewritten as $T_xL'\subset T_xL^{\perp_\bEta}$. Then, $0\neq T_xL'\cap W_x\subseteq T_xL^{\perp_\bEta}\cap W_x^{\perp_\bEta}=0$. This is a contradiction, and therefore $L\subset L'$ with $\dim L<\dim L'$ is not possible.
\end{proof}

Note that it is an immediate consequence of Theorem \ref{Th:MaxLeg} that a Legendrian subspace is maximal relative to isotropic subspaces.

The following proposition, namely Proposition \ref{prop:GeneratingFunction}, is related to the generalization of the notion of generating function of a Legendrian submanifold from contact to the $k$-contact setting. Moreover, it shows that while Legendrian submanifolds of $(2n+1)$-dimensional contact manifolds are always of dimension $n$, one has that $k$-contact manifolds admit Legendrian submanifolds of different dimensions.

In the formulation and the proof of the following theorem, we do not assume the Einstein summation convention to clarify the meaning of the ranges of the sums over indices.  

\begin{proposition}\label{prop:GeneratingFunction}
    Let $(M,\bEta,\mathcal{V})$ be a polarized $k$-contact manifold with $\dim M=k+nk+n$. Let $\{s^\alpha,q^l,p^\alpha_l\}$ be Darboux coordinates defined on an open subset $U\subset M$. Consider a disjoint partition $I \sqcup J$ of the set of indices $\{1,\ldots,n\}$ and  a set of functions
    \begin{equation}\label{eq:LegrendrePara}
    F^\alpha=F^\alpha(q^j, p^\alpha_i),\qquad \alpha=1,\ldots,k,\quad i\in I,\quad j\in J.
    \end{equation}
    Denote by $|I|=n_1$ and $|J|=n_2$ the cardinalities of the sets $I$ and $J$.    If
    \begin{equation}\label{eq:GeneratingFunctionCondition}
        \frac{\partial F^\alpha}{\partial p_i^\alpha} \equiv \frac{\partial F^\beta}{\partial p_i^\beta} ,\qquad  i \in I, \quad   \alpha, \beta =1, \dots, k,
    \end{equation}
    then the parametrization 
    $
    \varphi:\R^{n_2+kn_1}\longrightarrow U
    $
    given by 
    \begin{equation}\label{eq:LegendreParametrisation}
        (q^j, p_i^\alpha) \longmapsto \left(F^\alpha - \sum_{i \in I} p_i^\alpha \frac{\partial F^\alpha}{\partial p_i^\alpha},q^j, -\frac{\partial F^\alpha}{\partial p_i^\alpha},\frac{\partial F^\alpha}{\partial q^j},p_i^\alpha\right)
    \end{equation}
    defines a Legendrian submanifold $L\subset M$ of dimension $\dim L=n_2+kn_1=n+(k-1)n_1$.
\end{proposition}
\begin{proof}
    To prove that $L$ is isotropic, it suffices to calculate the pull-back of the $k$-contact form $\bEta$ defined on $M$ by the parametrization $\varphi:\R^{n_2+kn_1}\longrightarrow U$. Indeed, one has
    \begin{align*}
\varphi^* \eta^\alpha &= d\left(F^\alpha - \sum_{i\in I}p_i^\alpha \frac{\partial F^\alpha}{\partial p_i^\alpha}\right) + \sum_{i\in I}p_i^\alpha d\left(\frac{\partial F^\alpha}{\partial p_i^\alpha}\right) - \sum_{j\in J}\frac{\partial F^\alpha}{\partial q^j} dq^j \\
&= \sum_{j\in J}\frac{\partial F^\alpha}{\partial q^j} dq^j + \sum_{i\in I}\frac{\partial F^\alpha}{\partial p_i^\alpha} dp_i^\alpha - \sum_{i\in I}\frac{\partial F^\alpha}{\partial p_i^\alpha} dp_i^\alpha - \sum_{i\in I}p_i^\alpha d\left(\frac{\partial F^\alpha}{\partial p_i^\alpha}\right) + \sum_{i\in I}p_i^\alpha d\left(\frac{\partial F^\alpha}{\partial p_i^\alpha}\right) - \sum_{j\in J}\frac{\partial F^\alpha}{\partial q^j} dq^j = 0
\end{align*}
for $\alpha=1,\ldots,k$. 
Contrary to other parts of this work, the equation above does not assume the Einstein summation convention to clarify that every index ranges over different values. Since $L$ is a submanifold and its tangent space is included in $\ker \bEta$, one obtains that 
$$
(\varphi^*d\bEta)(X_1,X_2)=X_1\varphi^*\bEta(X_2)-X_2\varphi^*\bEta(X_1)-\varphi^*\bEta([X_1,X_2])=0,
$$
where we used that the commutator of $[X_1,X_2]$ is tangent to $L$ and $TL\subset \ker\bEta$ because $\varphi^*\bEta=0$. 

To prove that $L$ is Legendrian, it is still left to construct an appropriate supplementary  $W_x\subset \ker \bEta_x$ to $T_xL$ so that $d\bm \eta|_{W_x\times W_x}=0$ for every $x\in L$. If we set on $U$ a distribution
$$
W=\left\langle\frac{\partial}{\partial q^i},\frac{\partial}{\partial p^\alpha_j}\right\rangle,
$$
the rank of $W$ is complementary to the dimension of $TL$ in $\ker\bm\eta$ and $d\bm\eta|_{W\times W}=0$. Moreover, the projection $\tau:(s^\alpha,q^l,p^\alpha_l)\in U\mapsto (q^j,p^\alpha_i)\in \mathbb{R}^{n_2+kn_1}$ is such that $W\subset \ker T\tau$ while $T\tau\circ T\varphi=T{\rm Id}$. Hence, $W_x\oplus T_xL=\ker \bm \eta_x$ for every $x\in L$. 

Parametrization \eqref{eq:LegendreParametrisation} is well defined if $q^i=-\frac{\partial F^\alpha}{\partial p^\alpha_i}$ for all $\alpha=1,\ldots,k$, which amounts to condition \eqref{eq:GeneratingFunctionCondition}.
\end{proof}

\begin{note}
Condition \eqref{eq:GeneratingFunctionCondition} imposed on the set of functions $F^1,\ldots,F^k$ is very restrictive. Indeed, the equation
    $$
    \frac{\partial F^\alpha}{\partial p_i^\alpha}
    (q^j,p_l^\alpha) = \frac{\partial F^\beta}{\partial p_i^\beta} (q^j,p_l^\alpha),\qquad i,l\in I,\quad j\in J,\quad \alpha,\beta=1,\ldots,k,
    $$
    can be satisfied if and only if
    $$
    \frac{\partial F^\alpha}{\partial p_i^\alpha}(q^j,p_l^\alpha) = \frac{\partial F^\alpha}{\partial p_i^\alpha}(q^j) = \frac{\partial F^\beta}{\partial p_i^\beta}(q^j) = \frac{\partial F^\beta}{\partial p_i^\beta}(q^j,p_l^\alpha),\qquad i,l\in I,\quad j\in J,\quad \alpha,\beta=1,\ldots,k,
    $$
    so functions $F^\alpha$ must take the form
    $$
    F^\alpha(q^j,p_i^\alpha) = \sum_{i \in I} p_i^\alpha f^i(q^j),\qquad j\in J,
    $$
    for some functions $f^i\in C^\infty(M)$. This means that the $0=F^\alpha-\sum_{i \in I} p_i^\alpha f^i(q^j)$ in parametrization \eqref{eq:GeneratingFunctionCondition} and the coordinates $s^1,\ldots,s^k$ of the points of the Legendrian submanifold vanish. 
\end{note}
    
    In view of the above, the restriction of the momenta $p^\alpha_j$ on $L$, with $j\in J$, are linear combinations of the restriction to $L$ of the momenta $p_i^\alpha$ for $i\in I$ with coefficients depending only on the $q^j$ with  $j\in J$. This shows the special character of the above parametrisation for certain $k$-contact Darboux coordinates, which is anyhow justified for our posterior theory in  relativistic thermodynamics. Despite that, the following proposition illustrates which parts of the above parametrisation can be generalised to any Legendrian submanifold in $k$-contact geometry. 

\begin{proposition} \label{prop:Legendrian Parametrisation} Every Legendrian submanifold of a polarised $k$-contact manifold $(M,\bm \eta,\mathcal{V})$ admits a parametrisation in terms of a system of variables $\{q^i,p_j^\alpha\}$, where $i=1,\ldots,s$ while $j=s+1,\ldots,n$ and $\alpha=1,\ldots,k$, of a $k$-contact Darboux coordinates system $\{q^l,p_l^\alpha\}$ for $(M,\bEta,\mathcal{V})$.
\end{proposition}
\begin{proof} Consider the codistribution 
 $$
 \mathcal{D}_x=\{\iota_vd{\eta}_x^\alpha , v\in \mathcal{V}_x,\alpha=1,\ldots,k\},\qquad x\in M.
 $$ 
 Since $\bm \eta$ admits a polarisation, $\mathcal{D}$ has rank $n$ everywhere, as follows from the $k$-contact Darboux theorem. Let us define
$$
 \mathcal{E}_x=\{\iota_vd{\eta}_x^\alpha|_{TL}, v\in \mathcal{V}_x,\alpha=1,\ldots,k\},\qquad x\in L.
 $$ 
 Let $s$ be the maximal rank of $\mathcal{E}$ on an open subset $U$ of $L$, which exists because the rank of $\mathcal{E}$ is a lower semi-continuous function on $L$ (cf. \cite{VAI_94a}). 
 As $\mathcal{E}= \mathcal{D}|_{TL}$, one obtains that $s\leq n$ and $\mathcal{D}$ admits at every $x\in U\subset L$, some $n-s$ covectors that annihilate $T_xL$. The forms $d\eta_x^1,\ldots,d\eta_x^k$ restricted to $\ker\bEta_x$ become a $k$-symplectic linear form, which is non-degenerate,  on a linear space of dimension $(n+1)k$ with a polarisation. The linear $k$-symplectic Theorem in \cite[Theorem 3.2]{GLRR_24} can be applied to $d\bEta|_{\ker \bEta\times \ker \bEta}$. Using this theorem, one obtains a basis $e^1,\ldots,e^s,e^{s+1},\ldots,e^n$ of covectors such that  $e^{s+1},\ldots,e^n$ vanish on $T_xL+ \mathcal{V}_x$. Hence, the linear $k$-symplectic Theorem shows that at $x\in U\subset L$, one can write $d{\rm \eta}_x^\alpha=\sum_{l=1}^{n}e^l\wedge e_l^\alpha$ for $\alpha=1,\ldots,k$. If  $v\in T_xL$, then $\iota_vd{\rm \eta}_x^\alpha=\sum_{k=1}^s\langle e^k,v\rangle e_k^{\alpha}-\sum_{l=1}^ne^l\langle e^\alpha_l,v\rangle$.  Hence, 
 there exist $(n-s)k$ vectors in $\mathcal{V}_x$, namely $e_\alpha^{j}$ for $j=s+1,\ldots,n$ and $\alpha=1,\ldots,k$, that are orthogonal to every vector field taking values in $T_xL$ relative to $d\eta_x^1,\ldots,d\eta_x^k$. Since $L$ is Legendrian, its tangent space at $x$ is maximally isotropic, and the tangent vectors $e^j_\alpha$ for $\alpha=1,\ldots,k$ and $j=s+1,\ldots,n$  must be tangent to $L$, which must have, at least, dimension $s+k(n-s)$. This dimension is the largest possible, as more tangent vectors would imply that $e^1,\ldots,e^s$ are linearly dependent on $T_xL$, which goes against our initial assumption. Choose coordinates $q^1,\ldots,q^n$ that are first integrals of the vector fields taking values in  $\mathcal{V}$. It is a standard fact from differential geometry that one can choose them so that $dq^1_x=e^1,\ldots,dq^s_x=e^s$. Hence, $dq^1,\ldots,dq^{s}$ are functionally independent on a neighbourhood of $x\in L$ and $q^1,\ldots,q^s$ become functionally independent variables even when restricted to $L$ on a neighbourhood of $x$ within $L$. Note that Reeb vector fields are not tangent to $L$ and are Lie symmetries of $d\bEta$ and $\mathcal{V}$, which can be projected along to $L$ to a $k$-symplectic manifold with a polarization (at least locally).
 Then, Lemma 3.4 and Theorem 3.5 in \cite{GLRR_24} give rise to a system of coordinates  $\{q^l,p_l^\alpha\}$, where $l=1,\ldots,n$ and $\alpha=1,\ldots,k$. The $k$-contact Darboux theorem shows that the $d\eta^\alpha$ can be locally written as $\sum_{l=1}^n\sum_{\alpha=1}^kdq^l\wedge dp_l^\alpha$. From the maximal isotropy of $L$ and since $dq^{s+1},\ldots,dq^n$ vanish on $T_xL$, one finds that the tangent vectors $\partial/\partial p^\alpha_j$ for $j=s+1,\ldots,n$ must be tangent to $L$ at $x\in L$. Hence, the forms $dp^\alpha_j$ are functionally independent on a neighbourhood of $x\in L$.    We also know that $dq^1,\ldots,dq^s$ are functionally independent on $L$ and vanish on $\mathcal{V}$. Hence, $\{q^i,p^\alpha_j\}$ form a coordinate system when restricted to $L$.
\end{proof}

 \begin{note}
     It follows from Proposition \ref{prop:Legendrian Parametrisation} that the only possible dimensions of Legendrian submanifolds of polarised $k$-contact manifolds of dimension $k+nk+n$ are $ks$ for $s=1,\ldots,n$.
 \end{note}

Transformation \eqref{eq:LegendreParametrisation} for $k=1$ is the particular case of Arnold's parametrization of contact Legendrian submanifolds \cite[p. 367-368]{Ar89} of the form
$$
(\hat{q}^i,\hat{p}_j)\mapsto \left(F-\hat{q}^i\frac{\partial F}{\partial \hat{q}^i} ,\hat{q}^i,-\frac{\partial F}{\partial \hat{p}_j},\frac{\partial F}{\partial \hat{q}^i},\hat{p}_j\right)
$$
under the change $\{q^j=\hat{p}_j,p_i=-\hat{q}^i\}$ for $j\in J$ and $i\in I$. All these facts will have a physical meaning in the following sections. 

Note that the $s$ coordinate of the parametrisation  \eqref{eq:LegendreParametrisation} for $k=1$ can be understood as a Legendre transformation from the function $F(q^j,p_i)$ to a new function $G(q^j,q^i)$ provided that the $p_i$ can be written as functions of $\{q^j,q^i\}$. In such a case, one could understand that the image of the parametrization is of the form
$$
(q^j,q^i)\mapsto \left(G,q^j,q^i,\frac{\partial G}{\partial q^j},\frac{\partial G}{\partial q^i}\right),
$$
which is the first-jet prolongation of the function $G:\mathbb{R}^n\rightarrow \mathbb{R}$. In this sense, the parametrization \eqref{eq:LegendreParametrisation} is related to a natural local generating function for $L$, namely a function $G:(q^1,\ldots,q^n)\in \mathbb{R}^n\mapsto G(q^1,\ldots, q^n)\in \mathbb{R}$ such that the points of the Legendrian manifold are of the form $(G(q^l),q^l,d_qG)$. A similar construction can be achieved for $I=\emptyset$, $J=\{1,\ldots,n\}$,  and any value of $k$. 

In what follows, the function $F:\mathbb{R}^{n_2+kn_1}\mapsto (F^1,\ldots,F^k)\in \mathbb{R}^k$ is called the {\it parametrizing $k$-function} of the Legendrian submanifold $\varphi(\mathbb{R}^{n_2+kn_1})$. 

\section{{\it k}-contact approach to ``extensive relativistic hydrodynamics''}
Relativistic hydrodynamics is a phenomenological model of a classical field theory in which one considers the evolution of a fluid element in flat spacetime. This theory is commonly used in fields such as astrophysics and for describing quark-gluon plasma (QGP), and it has attracted significant interest in recent times in the high-energy physics community \cite{FHS18, RR07, Ro10, JR16}. Let us prove for the first time that relativistic hydrodynamics possesses a rich mathematical and geometrical structure, a rigorous formulation of which could provide significant insight into its nature.

Let us briefly sketch the mathematical description of contact thermodynamics, which is a basis for the development of $k$-contact formalism for relativistic hydrodynamics. The details of a contact description of thermodynamical systems can be found in \cite[Chapter 5]{Br19}. Consider a thermodynamical phase space given by a manifold $\R^7$ with canonical coordinates $(E, P, V, T, S,\mu, N)$. The first law of thermodynamics is frequently expressed as
\begin{equation}\label{eq:ILawThermo}
    dE - TdS - \mu dN + PdV=0,
\end{equation}
where we understand that $E,T,\mu$ and $P$ are functions depending on $S,N,V$. From a mathematical point of view, the equation above can be interpreted as the vanishing of the contact form
\begin{equation}\label{eq:ThermoContactForm}
    \eta=dE - TdS - \mu dN + PdV,
\end{equation}
when restricted to a submanifold describing a thermodynamical state. Indeed, equilibrium states of a thermodynamic system are certain submanifolds $L\subset \R^7$ with the property that $\eta|_{TL}=0$. For physical reasons, one requires the submanifold $L$ to be not only isotropic but also Legendrian (there is no other isotropic submanifold containing $L$). This ensures that the dimension of the Legendrian submanifold, in our system, is three, and this matches the physical degrees of freedom of a thermodynamic system in our setup.
Note that there exist several possible variables to describe our thermodynamic state. In one of the possible coordinate systems, $S, V, N$ are independent variables on $L$. Then, one may write on $L$ that 
$$
E=f(S,V,N).
$$
This is the equation of state of a thermodynamic system. From the perspective of contact geometry, the equation of state can be interpreted as a choice of a function $f$ that determines uniquely a Legendrian submanifold of the contact manifold $(\R^7,\eta)$. The Legendrian submanifold defined by the generating function $f(S, V, N)$ is given by a parametrisation
\begin{equation}\label{eq:ThermoParametrisation}
    (S, V, N)\longmapsto\left(f(S,V,N), S,V,N,\frac{\partial f(S,V,N)}{\partial S},-\frac{\partial f(S,V,N)}{\partial V},\frac{\partial f(S,V,N)}{\partial N}\right).
\end{equation}
Then, in an equilibrium state, which is understood as a point of a Legendrian submanifold, it holds
\begin{equation}\label{eq:IdealGasEqState}
    P=-\frac{\partial f(S,V,N)}{\partial V},\qquad T=\frac{\partial f(S,V,N)}{\partial S},\qquad \mu=\frac{\partial f(S,V,N)}{\partial N}.
\end{equation}
Moreover, assuming that $f(S, V, N)$ is a homogeneous function, what physically corresponds to a choice of an extensive thermodynamic system,  Euler's homogeneity theorem yields the Gibbs equality
\begin{equation}\label{eq:GibbsEq}
    E + PV = TS + \mu N,
\end{equation}
 on the chosen Legendrian submanifold.

A thermodynamic process may be modelled within contact Hamiltonian mechanics \cite{Br19,BLN15,MNSS91,Mr00,EMS07,Qu07}. In this picture, a process is an integral curve $\gamma(t)$ of the contact Hamiltonian vector field $X_H$ generated by a Hamiltonian $H\in C^\infty(M)$. To ensure that the evolution connects equilibrium states, one requires that $\gamma$ is contained in a Legendrian submanifold $L$, i.e. that $X_H$ should be tangent to $L$; this condition guarantees that the first law of thermodynamics is preserved along the process. A standard result in contact geometry (see \cite[Thm.~2.3]{Br19}) states that $X_H$ is tangent to $L$ if and only if $H\big|_L=0$. We illustrate this construction with an example.

\begin{example}
    Let us analyse an isentropic process of the ideal gas. The thermodynamic phase space of the ideal gas is the contact manifold $(\R^7,\eta)$ described at the beginning of this section. The equation of state reads
    $$
    S=Nc_VR\log(U/U_0)+NR\log(V/V_0)+Ns_0,
    $$
    where $R$ is the universal gas constant, $c_V$ is the specific heat and $U_0, V_0,s_0$ are constants. Therefore, the internal energy can be expressed as
    $$
    U=U_0\left(\frac{V}{V_0}\right)^{-1/c_V}e^{(S/N-s_0)/(c_VR)}.
    $$
    The Hamiltonian of an equilibrium isentropic process given by a Legendrian submanifold defined by $f=f(S,V,N)$ can be written in the form
    $$
    H_S:=-\left(P+\frac{\partial f}{\partial V}\right)V.
    $$
    The Hamiltonian was defined in such a way because $H_S$ is a Hamiltonian of an equilibrium process if and only if $H_S$ vanishes when restricted to a Legendrian submanifold given by the parametrisation  \eqref{eq:ThermoParametrisation}. The Hamiltonian vector field can be expressed in coordinates as
    $$
    X_{H_S}=X^U\partial_U+X^T\partial_T+X^S\partial_S+X^P\partial_P+X^V\partial_V+X^\mu\partial_\mu+X^N\partial_N.
    $$
    Then, the contact Hamiltonian equations, namely the geometric HdDW equations \eqref{eq:HdDW} for $k=1$, yield
    \begin{align*}
        & X^S=0,\quad X^T=V\frac{\partial^2f}{\partial V\partial S},\quad X^P=-P-\frac{\partial f}{\partial V}-V\frac{\partial^2 f}{\partial V^2},\quad X^V=V,\\ &X^N=0,\quad X^\mu=V\frac{\partial^2f}{\partial V\partial N},\quad X^U=\frac{\partial f}{\partial V}V.
    \end{align*}
    For a general Hamiltonian one obtains that $X^S=\frac{\partial H_S}{\partial T}$, so a Hamiltonian for an isentropic process must be independent of $T$. The isentropic process is an integral curve of the vector field $X_{H_S}$.
    
\end{example}

To create equations describing relativistic hydrodynamics, these equations are typically recast in a covariant (four-vector) form \cite{FH25}. This leads to describe the extensive entropy four-current $\mathbb{S}^\mu$ via the following expressions
\begin{align}
\mathbb{S}^\mu &= P^\mu V - \xi \mathbb{N}^\mu + \beta_\lambda \mathbb{T}^{\lambda\mu},\label{eq:4GibbsEquality} \\
d\mathbb{S}^\mu &= - \xi d\mathbb{N}^\mu + \beta_\lambda d\mathbb{T}^{\lambda\mu} + P^{\mu}dV.\label{eq:4ILawThermo}
\end{align}
Here, we use the standard notation $\beta^\mu = \frac{u^\mu}{T}$, $P^\mu = P\beta^\mu$, and $\xi = \frac{\mu}{T}$, where $u^\mu$ stands for the four-velocity of the fluid. The tensors $\mathbb{N}^\mu$ and $\mathbb{T}^{\lambda\mu}$ are the extensive baryon current and the extensive energy-momentum tensor, respectively. It is worth noting that there exist more general models, where the four-vector $P^\mu$ does not depend on $\beta^\mu$ and $P$, but it is an independent field \cite[p. 3]{FH25} - where one wants to incorporate spin degrees of freedom. It is important to notion that we use this ``extensive” formulation, which includes volume, for its consistency with the $k$-contact geometrical approach. 

In this case, the thermodynamic phase space is $\R^{4k+2+k^2}$ with coordinates $(\mathbb{S}^\mu, P^\mu, V, \xi, \mathbb{N}^\mu,\beta^\mu, \mathbb{T}^{\lambda\mu})$. In the extensive form of relativistic hydrodynamic relations, we easily recognize $k$-contact extensions of contact geometric structures employed in classical thermodynamics. Equation \eqref{eq:4ILawThermo} says that the space of equilibrium states should be a four-contact isotropic submanifold of the four-contact structure given by a four-contact form
\begin{equation}\label{eq:RelHydroKContactStructure}
    \bEta =(d\mathbb{S}^\mu + \xi d\mathbb{N}^\mu - \beta_\lambda d\mathbb{T}^{\lambda\mu} - P^{\mu}dV)\otimes e_\mu.
\end{equation}
Then, $\dim(\R^{k^2+2+4k})=k+nk+n$ for $n=k+2$, so the $k$-contact manifold $(\R^{k^2+2+4k},\bm\eta)$ may admit a polarization. Although such coordinates are not an example of $k$-contact Darboux coordinates, one can observe that the $k$-contact form defined by \eqref{eq:RelHydroKContactStructure} admits a polarization 
$$
\mathcal{V}=\left\langle\beta^\lambda\frac{\partial}{\partial \mathbb{S}^\mu}+\frac{\partial}{\partial \mathbb{T^{\lambda\mu}}},\xi\frac{\partial}{\partial \mathbb{S}^\mu}+\frac{\partial}{\partial \mathbb{N}^\mu}, \frac{\partial}{\partial P^\mu}\right\rangle
$$
of rank $nk=k(k+2)$, and one can construct Darboux coordinates for this four-contact form using the procedure described in the proof of existence of Darboux coordinates \cite[Theorem 7.1.11]{Ri21}. However, to preserve a clear physical interpretation of our equations, we do not perform such a coordinate change. It is worth stressing that the previous approach, and following parts of this work, work for $k=2,3$ so as describe thermodynamical models on space times with one or two spatial variables.

In Darboux coordinates $(s^\mu,p^\mu_i,q^i)$, one can construct Legendrian submanifolds which describe the equilibrium states using the procedure presented in Proposition \ref{prop:GeneratingFunction}. If additionally generating functions are extensive, one recovers the identity \eqref{eq:4GibbsEquality}, which is valid on the Legendrian submanifold similarly as in the contact case.

As a classical field theory, relativistic hydrodynamics is constructed from the fundamental fields that are maps from the spacetime $\R^{1,k-1}$ \footnote{The $k=4$ case corresponds to the Minkowski space time $\mathfrak{M}^{1,3}=\mathbb{R}^{1,3}$.} to the configurational space $\R^{k^2+4k+2}$. In our example a general form of a field reads
\begin{equation}\label{eq:HydroField}
    \psi(x)=(\mathbb{S}^\mu(x), P^\mu(x), V(x), \xi(x), \mathbb{N}^\mu(x),\beta^\mu(x), \mathbb{T}^{\lambda\mu}(x)),
\end{equation}
where $x$ stands for a point in $\R^{1,k-1}$. The equations of relativistic hydrodynamics proposed by physicists are the conservation laws for the energy-momentum tensor and the baryon current,

\begin{equation}
\partial_\mu \mathbb{T}^{\mu\nu}(x)=0, \hspace{1cm} \partial_\mu \mathbb{N}^{\mu}(x)=0.
\end{equation}
We will henceforth omit writing the dependence on $x$ for simplicity. A consequence of these conservation laws, combined with the equations \eqref{eq:4GibbsEquality} and \eqref{eq:4ILawThermo}, is the conservation of entropy, $\partial_\mu \mathbb{S}^{\mu}=0$, which implies the system is in thermal equilibrium \cite{DFHR24}. Let us try to derive these equations using the $k$-contact HdDW equations. The section \eqref{eq:HydroField} induces  $\psi'_{\mu}$ given by

\begin{equation}
    \psi'_{\mu}(x) =  \partial_\mu \mathbb{S}^{\nu}\partial_{\mathbb{S^{\nu}}} + \partial_{\mu}\xi \partial_{\xi} + \partial_{\mu} \mathbb{N}^{\nu}\partial_{\mathbb{N}^{\nu}} + \partial_{\mu}V \partial_{V} + \partial_{\mu}P^{\nu}\partial_{\beta_{\nu}} + \partial_{\mu}\mathbb{T}^{\alpha \beta} \partial_{\mathbb{T}^{\alpha \beta}}.
    \label{eq:psi_prim_revised}
\end{equation}

Let us now analyse the $k$-contact HdDW equations for the $k$-contact form
\begin{equation*}
{\bm \eta} = (d\mathbb{S}^{\mu} + \xi d\mathbb{N}^\mu - P^{\mu}dV - \beta_{\lambda}d\mathbb{T}^{\lambda \mu})\otimes e_\mu.
\end{equation*}
At first, we compute
\begin{equation*}
d{\bm \eta} = (d\xi \wedge d\mathbb{N}^{\mu} - d P^{\mu} \wedge dV - d\beta_{\lambda} \wedge d \mathbb{T}^{\lambda\mu})\otimes e_\mu.
\end{equation*}
Assuming a zero Hamiltonian ($H=0$), the HdDW equations for  $\psi'_\mu$ from equation  \eqref{eq:psi_prim_revised} are
\begin{equation}\label{eq:HdDW_H_0_revised}
\iota_{\psi'_\mu}d{\eta^{\mu}}= 0, \qquad
 \iota_{\psi_\mu'}{\eta^{\mu}}=0.
\end{equation}
The first equation expands to
\begin{equation*}
    \iota_{\psi'_\mu}d{\eta^{\mu}} = (\partial_\mu \xi) d\mathbb{N}^{\mu} - (\partial_\mu \mathbb{N}^{\mu}) d\xi - (\partial_\mu P^\mu)dV + (\partial_\mu V) dP^{\mu} - (\partial_\mu \beta_{\lambda})d\mathbb{T}^{\lambda \mu} + (\partial_\mu \mathbb{T}^{\lambda \mu}) d\beta_{\lambda} = 0.
\end{equation*}
For this equation to hold, the coefficients of the independent basis one-forms must vanish independently. This directly yields the conditions $\partial_\mu \xi = 0$, $\partial_\mu \mathbb{N}^{\mu} = 0$, $\partial_\mu P^{\mu} = 0$, etc. The second  equation in \eqref{eq:HdDW_H_0_revised} is
\begin{equation}
    \iota_{\psi'_\mu}{\eta^{\mu}} = \partial_\mu \mathbb{S}^\mu + \xi (\partial_\mu \mathbb{N}^{\mu}) - P^{\mu}(\partial_\mu V) - \beta_\lambda (\partial_\mu \mathbb{T}^{\lambda \mu}) = 0.
\end{equation}
Using the results from the first equation (e.g., $\partial_\mu \mathbb{N}^{\mu}=0$), this simplifies to $\partial_\mu \mathbb{S}^\mu = 0$, confirming that entropy production is zero. Collecting all these conditions, we obtain the full description of the process:
\begin{align}
    &\partial_\mu \xi = 0,      && \partial_\mu \mathbb{N}^{\mu} = 0, && \partial_\mu P^{\mu} = 0, \nonumber \\
    &\partial_\mu V = 0,      && \partial_\mu \beta_\lambda = 0,   && \partial_\mu \mathbb{T}^{\mu\nu} = 0, && \partial_\mu \mathbb{S}^{\mu} = 0.
    \label{eq:HdDW_Phys_1_revised}
\end{align}
These equations describe a system in global thermodynamic equilibrium. The conditions on $\xi$ and $\beta^\lambda$, which are analogous to generalized Tolman-Klein conditions for thermal equilibrium, enforce a uniform flow with no temperature or chemical potential gradients. As expected for an equilibrium system, all extensive currents are conserved, and entropy production is zero.

\section{Analysis of pseudo-gauge degrees of freedom}

In the context of relativistic hydrodynamics, pseudo-gauge transformations are defined as modifications to the energy-momentum tensor, $\mathbb{T}^{\lambda\mu}$, and the baryon current, $\mathbb{N}^\mu$, that do not violate the fundamental conservation equations \eqref{eq:HdDW_Phys_1_revised}. However, the physical interpretation of this pseudo-gauge ambiguity remains unclear. It is known that while total quantities are pseudo-gauge invariant, quantum energy density correlations are dependent on the specific pseudo-gauge choice \cite{DFRS21}.

Furthermore, not every system can undergo a pseudo-gauge transformation due to an insufficient number of degrees of freedom \cite{DFHR24}. For a general configuration in causal hydrodynamics, demanding that the energy-momentum tensor remains symmetric both before and after the transformation leads to an equation for a super-potential $\Phi^{\lambda\mu\nu}$. This equation is generally unsolvable. Only by imposing additional symmetry assumptions on the system, such as the Bjorken flow \cite{DFHR24}, can the equations be solved, thus permitting the pseudo-gauge transformation.

This dependence presents a significant conceptual challenge in theories where physically measurable quantities rely on the explicit form of the pseudo-gauge. Examples include the density operator used to compute averages in Zubarev’s approach to relativistic hydrodynamics \cite{BBG19, BH25}, viscosity as derived from the Kubo formula \cite{HST83}, and other transport coefficients within the BDNK formalism \cite{Ko19}. This implies one of two possibilities: either the choice of pseudo-gauge can modify physically measurable quantities, or every physical measurement is already performed in a specific, albeit arbitrary, pseudo-gauge. The latter proposition seems unnatural, as physical reality is expected to be unique and independent of the mathematical formalisms used to describe it.

From a mathematical point of view, such transformations are possible due to the non-uniqueness of solutions of the geometric $k$-contact HdDW equations that was analysed in Theorem \ref{thm:HdDWNonUniqueness}. To obtain a full description of a relativistic hydrodynamical system, one needs to ensure that solutions $\psi$ of the $k$-contact HdDW equations \eqref{eq:HdDW} take values in isotropic submanifolds. This implies that $\psi'_{\mu}$ is tangent to an isotropic submanifold. Let us prove two theorems that give when such solutions are possible.
\begin{proposition}\label{prop:ConstrainedHamiltonian}
    If the image of a solution to $k$-contact HdDW equations with a Hamiltonian $H$ is contained in an isotropic submanifold $L$, then $H\vert_L=0$.
\end{proposition}
\begin{proof}
    If the image of a solution $\psi$ of $k$-contact HdDW equations is an isotropic submanifold $L=\im(\psi)$, then the $k$-vector field $\psi'_\mu$ is tangent to $L$. Thus, from the second HdDW equation, we get
    $$
    -H\vert_L=\iota_{\psi'_\mu}\eta^\mu=0,
    $$
    as $L$ is isotropic.
\end{proof}
\begin{theorem}\label{th:ConstrainedDynamics}
    If $L\subset M$ is a Legendrian submanifold and $H\vert_L=0$, then the geometric $k$-contact HdDW equations admit solutions that are tangent to $L$.
\end{theorem}
\begin{proof}
Let us look for a solution $\bX$ of equations \eqref{eq:HdDW} such that $\bX$ can be considered as a section of $\oplus^k TL$ over $L$. For such a solution, 
$$
\iota_{X_\alpha}\eta^\alpha\vert_L=0,
$$
    and the second HdDW equation is satisfied on $L$. Note that every Lagrangian submanifold admits by definition an isotropic supplementary $W$. Since $W+\ker \bm \eta=TM$, it follows that $F=\ker \bm \eta\cap W$ is a supplementary to $L$ in $\ker\bm\eta$ that is isotropic relative to $d\bm\eta$. Hence, $F\oplus TL=\ker{\bm \eta}$ and the map 
\begin{align*}
    \rho_1:\bigoplus^k \ker{\bm \eta}&\longrightarrow (\ker{d\bm \eta})^\circ,\\
    \bX^C&\longmapsto\iota_{\bX^C}d\bEta,
\end{align*}
can be presented as
$
\rho_1:(\oplus^k TL)\oplus(\oplus^k F)\longrightarrow (TL)^*\oplus F^*,
$ where the elements of $(TL)^*$ annihilate $F$ and the elements of $F^*$ annihilate $TL$. 
Then, $\rho_1=\rho_L\oplus\rho_F$, where
$
\rho_L:\oplus^k TL\longrightarrow F^*,
$
and
$
\rho_F:\oplus^k F\longrightarrow (TL)^*.
$
By Theorem \ref{thm:HdDWNonUniqueness}, $\rho_1$ is surjective, so both $\rho_L$ and $\rho_F$ must be surjective. As $H\vert_L=0,$ then the right-hand side of the first  HdDW equation $dH-(\mathcal{L}_{R_\alpha}H)\eta^\alpha$ belongs to $F^*$. The surjectivity of $\rho_L$ guarantees the existence of a solution tangent to $L$.
\end{proof}
\begin{note} Recall that if we assume $\dim M=nk+k+n$, 
    a necessary condition for the map $\rho_L$ to be surjective is $k\dim L\geq \rk F=n(k+1)-\dim L$, which amounts to \(\dim L\geq n\). Therefore, solutions of geometric $k$-contact HdDW equations that are tangent to an arbitrary isotropic submanifold, whose dimension can be even zero, may not exist.
\end{note}

\begin{note}\label{no:LimitedPseudo-Gauge}
    Note that $\dim\ker\rho_L=k\dim L-\rk F\geq 0$ and $\dim\ker\rho_F=k\rk F-\dim L\geq 0$. Therefore, if $(M,\bEta,\mathcal{V})$ is a polarized $k$-contact manifold of dimension $\dim M=k+nk+n$, then,  $\dim L+\rk F=nk$ and  the dimension of a $k$-contact Legendrian submanifold of $M$ is always bounded by $n\leq \dim L \leq nk$. Moreover, the number of pseudo-gauge degrees of freedom is $\dim\ker\rho_L=k\dim L-\rk F$. Thus, a solution of the HdDW equations for a Hamiltonian satisfying $H\vert_L=0$ that is tangent to $L$ is unique only if $\dim L =n$.
\end{note}

Physically, Theorem \ref{th:ConstrainedDynamics} can be understood as follows. One starts from the same theoretical description in terms of conserved currents and thermodynamic equations, as given by the geometric HdDW equations on a $k$-contact manifold $(M,\bEta)$. Due to the non-uniqueness of the solution, one can perform a mapping between solutions that is a pseudo-gauge transformation. However, in addition to conservation laws described mathematically by geometric $k$-contact HdDW equations, one needs to satisfy equation \eqref{eq:4ILawThermo}. This equation is obeyed on the image of the solution $\phi:\R^{1,k-1}\longrightarrow M$ when is contained in a chosen isotropic submanifold $L$. If one assumes that $L$ is additionally Legendrian, then Theorem \ref{th:ConstrainedDynamics} guarantees the existence of solutions of geometric $k$-contact HdDW equations whose image is contained in $L$ under the assumption that $H\vert_L=0$. However, in general, the constraint $\im(\phi)\subseteq L$ decreases the number of pseudo-gauge degrees of freedom.

This, however, implies that the system depends on the dimension of the admissible Legendrian submanifold, and a situation may arise in which no degrees of freedom remain for performing a pseudo-gauge transformation. 
This allows us to understand why, historically, pseudo-gauge freedom—such as the Belinfante procedure—was first observed in systems with spin degrees of freedom \cite{Be39}. The presence of spin provides the necessary freedom to perform a pseudo-gauge transformation. Nevertheless, certain hydrodynamical systems, such as the Bjorken flow, may also admit pseudo-gauge transformations due to their symmetries.

\section{Physical Interpretation of the Pseudo-Gauge}

After studying the mathematical perspective on the pseudo-gauge, we now discuss how this ambiguity manifests in specific physical configurations. The physical significance of the pseudo-gauge, which is mathematically identified with an underdetermination of the HdDW equations' solutions, remains an open question. As it follows from Remark \ref{no:LimitedPseudo-Gauge}, the number of degrees of freedom of pseudo-gauge transformations preserving Legendrian submanifolds might be limited, so some pseudo-gauge transformations may transform solutions of HdDW equations between different Legendrian submanifolds. Such an operation can be interpreted as a change of thermodynamic system.

The key physical insight is that one does not modify the conservation equation $\partial_{\mu}T^{\mu\nu}=0$ itself. Instead, the energy-momentum tensor $T^{\mu\nu}$ is modified by adding a total derivative. A significant consequence of this, regardless of whether a classical or quantum approach is taken, is that physically meaningful quantities—such as energy density $\varepsilon$, shear viscosity $\eta$, and bulk viscosity $\zeta$—depend on the specific form of the energy-momentum tensor. This implies that a particular choice of pseudo-gauge would correspond to different physically measurable quantities, which presents a conceptual challenge. For comparison, in QCD and other field theories, it is a common principle that physically relevant quantities must be true gauge-independent. It would be preferable for the same to hold in our pseudo-gauge setup. Moreover, in the context of describing strongly interacting matter, the appearance of a new freedom (or ambiguity) is unexpected. We will now demonstrate with a simple example how the pseudo-gauge affects entropy calculations.

Hereafter, we use natural units ($\hbar = c = 1$). The metric tensor is in the ``mostly-minuses” convention, $g_{\mu\nu} = \operatorname{diag}(+1,-1,-1,-1)$. For the Levi-Civita tensor $\epsilon^{\mu\nu\alpha\beta}$, we follow the sign convention $\epsilon^{0123} = -\epsilon_{0123} = +1$. The projection operators $\Delta^{\mu \nu }$ and $\Delta^{\mu \nu}_{\alpha \beta} $ are defined by
\begin{equation}
\Delta^{\mu \nu } \equiv g^{\mu \nu} - u^\mu u^\nu, \qquad \Delta^{\mu \nu}_{\alpha \beta} \equiv \frac{1}{2} \left(\Delta^{\mu}_{\phantom{\mu}{\alpha}} \Delta^{\nu}_{\phantom{\nu}{\beta}} + \Delta^{\mu}_{\phantom{\nu}{\beta}} \Delta^{\nu}_{\phantom{\nu}{\alpha}} - \frac{2}{3}\Delta^{\mu \nu } \Delta_{\alpha \beta } \right),
\end{equation}
where $u^\mu$ is the four-velocity of the fluid. The physical approach starts by a given energy-momentum tensor $T^{\mu \nu}(x) \equiv T^{\mu \nu}$ and performing a transformation
\begin{equation}
    T'^{\mu \nu} = T^{\mu \nu} + \frac{1}{2}\partial_\lambda (\Phi^{\lambda \mu \nu } - \Phi^{\mu \lambda \nu } - \Phi^{\nu \lambda \mu } ),
\end{equation}
where $\Phi^{\lambda \mu\nu}$ is an arbitrary tensor of rank 3, anti-symmetric in its last two indices ($\Phi^{\lambda \mu \nu} = - \Phi^{\lambda \nu \mu}$). This tensor, often called {\it super-potential}, ensures that the transformed energy-momentum tensor is also conserved, namely $\partial_\mu T'^{\mu \nu}=0$. This transformation is what is meant by modifying the tensor ``by a total derivative”.

In conventional relativistic hydrodynamics, a pseudo-gauge transformation (PGT) must be constructed from the basic hydrodynamic fields: temperature $T(x)$, chemical potential $\mu(x)$, and four-velocity $u^\mu(x)$, where $x \in \mathbb{R}^{1, k-1}$. For a first-order theory, which contains no terms like $\partial_\mu T(x)$, the super-potential $\Phi^{\lambda \mu \nu}$ must be constructed from the fields themselves, not their derivatives. This is because applying a PGT with a derivative-dependent super-potential would introduce higher-order derivative terms, which should be neglected in a first-order physical description.

This physical choice of independent fields means that, in the $k$-contact manifold $\R^{k(k+2)+k+(k+2)}$ with coordinates $(\mathbb{S}^\mu, P^\mu, V, \xi, \mathbb{N}^\mu,\beta^\mu, \mathbb{T}^{\lambda\mu})$, we want to choose a Legendrian submanifold representing a set of equilibrium states of a physical system whose independent variables are $\beta^\mu, \xi, V$. In other words,  we choose a Legendrian submanifold of the smallest possible dimension $\dim L=n=k+2$, e.g. $n=6$ for $k=4$.
Here, $\beta^{\mu}=\frac{u^\mu}{T}$, and the four velocity $u^{\mu}$ is normalized to 1, meaning $u_\mu u^{\mu}=1$ as in standard relativity theory, and so $\beta_\mu \beta^{\mu}=\frac{1}{T^2}$.
Let us examine a model known as ``Bjorken flow”. This framework describes a simple, one-dimensional boost-invariant expansion but yields non-trivial results regarding the pseudo-gauge problem.

In the Bjorken-like model with no spin degrees of freedom, we assume that $P^\mu=P(T)u^\mu$, what can be written as $P^\mu=P(\beta^\lambda\beta_\lambda)\beta^\mu$.Moreover, we have assumed no baryon current, so $N^\mu=0$.
The energy-momentum tensor for such a model is a symmetric rank-2 tensor, which can be expressed as
\begin{equation*}
    T^{\mu\nu} = \mathcal{E}(T) u^\mu u^\nu - \left(\mathcal{P}_{\rm eq}(T) + \Pi \right) \Delta^{\mu\nu} + \mathcal{T}^{\mu\nu},
\end{equation*}
where $\mathcal{E}$ is the energy density, $\mathcal{P}_{\rm eq}$ is the equilibrium pressure, $\Pi$ is the dissipative part of the isotropic pressure (bulk pressure), and $\mathcal{T}^{\mu\nu}$ is the shear-stress tensor, given by $\mathcal{T}^{\mu\nu}= \Delta^{\mu \nu}_{\alpha \beta} T^{\alpha \beta}$.

The framework of $k$-contact geometry requires the use of extensive variables that include volume. We therefore introduce the ``total'' energy-momentum tensor,\footnote{Note that this approach \textbf{differs} from standard physical theories, which use densities. However, it is consistent with $k$-contact geometry and can be connected to the usual density-based framework via scaling symmetries. The latter case is of great interest as it could provide a strong claim on the nature of pseudo-gauge freedom in relativistic hydrodynamics.} which we denote as $\mathbb{T}^{\mu\nu}$ and require to be conserved. This would correspond to ``Extensive, Bjorken-like model”. By using the extensive variables for total energy $E = \mathcal{E}V$, total dissipative pressure $\Pi_{tot}$, and total shear tensor $\mathcal{T}^{\mu\nu}_{tot} $, the expression becomes
\begin{equation*}
    \mathbb{T}^{\mu\nu} = E(T)u^{\mu}u^{\nu} - (P_{eq}(T)V + \Pi_{tot})\Delta^{\mu\nu} + \mathcal{T}^{\mu\nu}_{tot}.
\end{equation*}

 Mathematically, this expression is just a set of $k^2$ functions of the variables $\beta^\lambda$, $T$ and $V$. To get a full description of a parametrization of a Legendrian submanifold, one has to use the Poincaré lemma to recover a formula for the entropy current $S^\mu$ from  equation \eqref{eq:4ILawThermo}. Since the chosen Legendrian submanifold is of the smallest possible dimension $\dim L=n$, it follows from Remark \ref{no:LimitedPseudo-Gauge} that any change of solutions of HdDW equations, including a modification of the extensive energy-momentum tensor $\mathbb{T}^{\mu\lambda}$, changes the Legendrian submanifold representing a set of equilibrium states too. The following calculations advocate an interpretation that, from a physical point of view, such a change indeed corresponds to a change (or redefinition) of the physical equilibrium state of the system.

The extensive entropy production is found from the projection of the conservation law, $u_\nu \partial_\mu \mathbb{T}^{\mu \nu}=0$. Using the first law of thermodynamics for extensive variables ($TDS_{tot}=DE_{eq}+P_{eq}DV$) at constant particle number (where $D \equiv u^\mu \partial_\mu$) and the Gibbs-Duhem relation, one finds the source for the total entropy
\begin{equation*}
    T \partial_\mu \mathbb{S}^{\mu}_{tot} = \mathcal{T}^{\mu\nu}_{tot}\sigma_{\mu\nu} - \Pi_{tot} \theta,
\end{equation*}
where $ \mathbb{S}^\mu_{tot}$ is the total entropy current, $\theta = \partial_\mu u^\mu$ is the expansion scalar, and $\sigma_{\mu\nu}$ is the shear tensor. For an ideal fluid, the dissipative terms $\mathcal{T}^{\mu\nu}_{tot}$ and $\Pi_{tot}$ are zero, and thus the total entropy production vanishes. We now apply a PGT to this initial state of a perfect fluid, $\mathbb{T}^{\mu\nu}=E_{eq}u^\mu u^\nu -P_{eq}V \Delta^{\mu\nu}$, which is in equilibrium with $\partial_\mu \mathbb{S}^{\mu}_{tot}=0$.

Following the analysis in \cite{DFHR24}, an extensive PGT generated by the super-potential 
\[
\Phi^{\lambda \mu \nu} = \gamma \, I ( u^{\mu} \Delta^{\lambda \nu} -u^\nu \Delta^{\lambda \mu}),
\]
where $\gamma$ is a constant with dimensions of volume and $I$ is a scalar function of temperature, modifies the components of the total energy-momentum tensor as
\begin{equation*}
E'_{eq} = E + \gamma \,I \theta, \qquad P'V = PV - \gamma DI - \frac{2 \gamma}{3}I \theta, \qquad \mathcal{T}'^{\mu \nu}_{tot} = \mathcal{T}^{\mu\nu}_{tot}- \gamma I \sigma^{\mu \nu}.
\end{equation*}
To ensure that the total entropy production remains invariant, we \textit{redefine the equilibrium state}. We define a new temperature $T'$ and new total equilibrium energy $E'_{eq}(T')$ and pressure $P'_{eq}(T')V$ such that they absorb the PGT terms
\begin{align*}
    E'_{eq}(T') &= E_{eq}(T) + \gamma \, I \theta, \\
    P'_{eq}(T')V &= P_{eq}(T)V - \gamma \, DI.
\end{align*}
This choice defines the total dissipative bulk pressure in the new frame as $\Pi'_{tot} = - \frac{2 \gamma}{3}I \theta$. The new total entropy production is then given by
\begin{equation*}
    T' \partial_\mu S'^{\mu}_{tot} = \mathcal{T}'^{\mu\nu}_{tot}\sigma_{\mu\nu} - \Pi'_{tot} \theta = - \gamma I \left(\sigma^{\mu\nu}\sigma_{\mu\nu} - \frac{2}{3}\theta^2 \right).
\end{equation*}
For Bjorken flow, a direct calculation shows that $\sigma_{\mu\nu}\sigma^{\mu\nu} = \frac{2}{3}\theta^2$. Therefore, the right-hand side of the equation vanishes identically. This demonstrates that for any choice of the scalar $I$, the total entropy production remains zero. The PGT is absorbed by a redefinition of equilibrium, leaving the physical dissipation invariant. This result is consistent with the $k$-contact geometry approach. Moreover, this redefinition of the equilibrium state finds a conceptual parallel in recent quantum statistical calculations by \cite{BH25}, where changes in the local thermodynamic equilibrium density operator under a pseudo-gauge transformation were studied.

A potential complication in this extensive-variable framework is the behaviour of the volume element $V$. While a condition like  $ \partial_\mu V = 0 $, obtained from the $k$-contact formalism, would simplify the analysis, it is not satisfied by physical expansions such as Bjorken flow. In our model, this issue is managed by the introduction of the dimensionfull? constant $\gamma$, which serves as a workaround. However, a more fundamental solution to integrating extensive variables into a field theory, making it ``extensive field theory”, can be achieved by generalizing?, a topic we will explore in a forthcoming publication.
Finally, we note that a less trivial situation with pseudo-gauge transformations arises when the system has no additional symmetries or when more degrees of freedom are allowed, as in spin hydrodynamics. These are intriguing topics for future investigation.

\section{Conclusions and outlook}

In this paper, we have introduced a novel geometric framework based on $k$-contact geometry to address the long-standing issue of pseudo-gauge ambiguity in relativistic hydrodynamics. By treating a hydrodynamic-like theory as a classical field theory with volume to a  $k$-contact manifold, we have provided a systematic, first-principles origin for the freedom in defining conserved currents. Our approach unifies the system's thermodynamics, encoded in a $k$-contact form $\bEta$, with its dynamics, described by the Hamilton-de Donder-Weyl (HdDW) equations and Legendrian submanifolds. Mathematically, our paper starts an intuitive and physically based analysis of Legendrian submanifolds in $k$-contact geometry and their generating functions. It provides a new approach to the definition of such spaces and some of the potential parametrisations. 

Our principal finding is that the pseudo-gauge freedom is not an ad hoc ambiguity but rather a direct consequence of the inherent non-uniqueness of solutions to the geometric $k$-contact HdDW equations for $k>1$, which physically corresponds to conservation equations and additional conditions for the system. We have proposed a clear physical interpretation for this mathematical feature: the choice of a particular solution $\psi_{\mu}'$ corresponds to the choice of a specific thermodynamic equilibrium  ``frame'' for the system. In this view, a pseudo-gauge transformation is equivalent to a redefinition of what constitutes the non-dissipative part of the energy-momentum tensor and other currents.

To demonstrate the practical implications of our formalism, we analysed a model of a Bjorken-like expansion. We explicitly showed how a pseudo-gauge transformation, which introduces apparent dissipative terms, can be fully absorbed into a redefinition of the equilibrium energy and pressure. Crucially, this redefinition left the physical dissipation—measured by the total entropy production—invariant. This result strongly supports our central thesis that PGTs correspond to a change of equilibrium frame rather than a change in the system's intrinsic dissipative properties.

Our work opens several avenues for future investigation. The present analysis was conducted within an ``extensive" formulation of hydrodynamics. A key challenge is to fully integrate this geometric picture with the standard, density-based formulation of field theory, which should be potentially possible by generalizing the Kirillov structures and through the careful implementation of scaling symmetries. Furthermore, applying this $k$-contact formalism to more complex scenarios is a natural next step. Investigating systems without the high degree of symmetry of the Bjorken flow, or incorporating additional degrees of freedom such as those in spin hydrodynamics or magnetohydrodynamics, will provide more stringent tests of our framework. Moreover, a theory of $k$-contact relativistic hydrodynamics with additional constraints is worth developing, as many physical models assume the symmetry of the energy-momentum tensor. Finally, exploring the connection between our classical, geometric redefinition of equilibrium and recent quantum statistical approaches could yield deeper insights into the fundamental nature of thermal states in relativistic field theory.

\medskip
\noindent
{\it Acknowledgments}

MH acknowledges partial financial support from the Polish National Science Centre under Grant No. 2022/47/B/ST2/01372. He is also grateful for the hospitality of the Faculty of Physics at the University of Warsaw and for the financial support from the Faculty of Physics, Astronomy and Applied Computer Science at Jagiellonian University. Finally, MH acknowledges that the study was funded by ``Research support module'' as part of the ``Excellence Initiative - Research University''
program at the Jagiellonian University in Kraków. AM acknowledges funding from the National Science Centre (Poland) within the project 
WEAVE-UNISONO, No. 2023/05/Y/ST1/00043.

\printbibliography

\end{document}